\documentclass[journal]{IEEEtran}

\usepackage{amsmath}
\usepackage{amsfonts}
\usepackage{amssymb}
\usepackage{amsthm}
\usepackage{tabularx}
\usepackage{mathrsfs} 

\usepackage{url}
\usepackage{hyperref}
\newtheorem{theorem}{Theorem}
\newtheorem{remark}{Remark}

\newtheorem{proposition}{Proposition}

\newtheorem{corollary}{Corollary}

\usepackage{stfloats}
\usepackage{float}
\usepackage{graphicx}
\usepackage{cite}
\usepackage{xcolor}
\usepackage{subfigure}
\renewcommand{\vec}[1]{\mathbf{#1}}

\makeatletter
\def\blfootnote{\xdef\@thefnmark{}\@footnotetext}
\makeatother

\begin{document}
	
		\title{\huge{RIS-aided Secure Communications over Fisher-Snedecor $\mathcal{F}$ Fading Channels}} 
	\author{Farshad~Rostami~Ghadi\IEEEmembership{}, Wei-Ping Zhu, and Diego Mart\'in}
	\maketitle
\begin{abstract}
	\textcolor{blue}{In} this paper, we investigate the performance of physical layer security (PLS) over reconfigurable intelligent surfaces (RIS)-aided wireless communication systems, where all fading channels are modeled with Fisher-Snedecor $\mathcal{F}$ distribution. Specifically, we consider a RIS with $N$ reflecting elements between the transmitter and the legitimate receiver to develop a smart environment and also meliorate secure communications. In this regard, we derive the closed-form expressions for the secrecy outage probability (SOP) and average secrecy capacity (ASC). \textcolor{blue}{We also analyze the asymptotic behaviour of the SOP and ASC by exploiting the residue approach}. Monte-Carlo (MC) simulation results are provided throughout to validate the correctness of the developed analytical results, showing that considering RIS in wireless communication systems has constructive effects on the secrecy performance.
	
\end{abstract}

\section{Introduction}\label{sec1}
Reconfigurable intelligent surfaces (RISs) have been recently offered as a cost-effective technique to enhance reliability and improve the performance of future wireless communications \cite{basar2019wireless}. Specifically, RISs are artificial surfaces that are made of electromagnetic material and have a large number of passive reflecting elements, which can intelligently control the wireless propagation environment and ameliorate the received signal quality. On the other hand, the explosive growth in using wireless smart devices and also the broadcast nature of wireless channels have made momentous challenges for security and privacy in designing future wireless networks such as sixth-generation (6G) technology. One flexible approach to protect information from unauthorized access and guarantee secure communication is physical layer security (PLS). Given that PLS, which is widely referred to as the classic wiretap channel \cite{wyner1975wire}, leverages the physical characteristics of the propagation environment, accurate modeling of the statistical characteristics of fading channel coefficients is very important. To this end, the Fisher-Snedecor $\mathcal{F}$ distribution has been recently introduced in \cite{yoo2017fisher} to correctly model the combined effects of shadowing and multipath fading in wireless device-to-device (D2D) communications. The main advantage of Fisher-Snedecor $\mathcal{F}$ distribution is to provide a better tail matching of the empirical cumulative density function (CDF) for composite fading compared with the Generalized-$\mathcal{K}$ model. Therefore, since the empirical CDF tail represents deep fading, the Fisher-Snedecor $\mathcal{F}$ distribution is one of the best practical scheme for modeling fading channels in wireless networks. In addition, the probability density function (PDF) of Fisher-Snedecor $\mathcal{F}$ only consists of elementary functions with respect to the random variable (RV), thereby it can provide a more tractable analysis than the Generalized-$\mathcal{K}$ model. 

In recent years, intense research activities have been carried out on the role of RIS in PLS from various aspects \cite{tang2021novel,trigui2021secrecy}, however, there are limited works in the context of RIS-aided secure communications over different fading channels \cite{yang2020secrecy,ai2021secure,yadav2022secrecy,makarfi2020physical,zhang2021physical}. For instance, the secrecy performance of RIS-aided communications under Rayleigh fading channels was investigated in \cite{yang2020secrecy}, assuming the eavesdropper also can receive signals from the RIS. By deriving the closed-form expression for the secrecy outage probability (SOP), the authors showed that the system with the RIS and without the RIS-eavesdropper link has the best secrecy performance, which means that the secrecy performance becomes worse when the eavesdropper also utilizes the advantage of the RIS. In \cite{ai2021secure}, assuming there is no direct link between transmitter and legitimate receiver, the authors studied secure vehicle-to-vehicle (V2V) communications under Rayleigh fading and derived a closed-form expression for the SOP in terms of the bivariate Fox's H-function. Furthermore, the authors in \cite{yadav2022secrecy,makarfi2020physical,zhang2021physical} studied RIS-aided secure communication systems under Nakagami-$m$ and \textcolor{blue}{Fisher-Snedecor} $\mathcal{F}$ fading channels, respectively. \textcolor{blue}{In particular, the authors in \cite{yadav2022secrecy} obtained analytical expressions for the SOP and the ergodic secrecy capacity under Nakagami-$m$ fading channels.} The authors in \cite{makarfi2020physical} derived the SOP and the average secrecy capacity (ASC) under \textcolor{blue}{Fisher-Snedecor} $\mathcal{F}$ fading channels for a specific scenario, where only the transmitter uses a RIS-based access point (AP) for transmission and there is no RIS node between the transmitter and legitimate receiver. The authors in \cite{zhang2021physical} also analyzed SOP and ASC over cellular communications, where they used the stochastic geometry theory to generate the PDF and CDF of
the received signal-to-interference-plus-noise ratio (SINR) at each nodes. \textcolor{blue}{By proposing a simple heuristic algorithm to find the optimal phase shift of each RIS element, the authors in \cite{nguyen2021secrecy} analyzed the SOP and the secrecy rate for RIS-aided communication systems under Rician fading channels. Furthermore, the authors in \cite{wei2022secrecy} derived the closed-form expression of the ergodic secrecy capacity for a RIS-aided communication system with spatially random unmanned aerial vehicles (UAVs) acting as eavesdroppers, where the fading channels follow the Rician/Rayleigh distribution.}

Motivated by the potential of RIS and PLS technologies in performance improvement of future wireless communication systems and the significant role of Fisher-Snedecor $\mathcal{F}$ distribution in accurate modeling and characterization of the simultaneous occurrence of multipath fading and shadowing, in this paper, we considered a RIS-aided secure communication system over Fisher-Snedecor $\mathcal{F}$ fading channels. Specifically, \textcolor{blue}{in} contrast to previous works which used stochastic geometry to provide the distributions of received signal-to-noise ratios (SNRs) in terms of complicated multivariate Fox's H-function or considered a specific scenario in using RIS, we first provide analytical expressions for the PDF and the CDF of SNRs in terms of Gamma distribution, and then we \textcolor{blue}{derive} the SOP and the ASC in closed-form expressions. \textcolor{blue}{Furthermore, we provide the asymptotic analysis of the obtained secrecy metrics with high accuracy by exploiting the residue approach given in \cite{chergui2016performance}.} Eventually, Monte-Carlo (MC) simulation results validate the accuracy of the analytical results. \textcolor{blue}{The results indicate that considering the RIS in secure communications can improve the secrecy metrics performance. In addition, exploiting the Fisher-Snedecor $\mathcal{F}$ distribution for the fading channels provides a more accurate channel model and secrecy analysis compared with other fading distributions since it takes into account the channel gain, which consists of both signal power and channel characteristics.} To the best of the authors' knowledge, our analytical expressions are novel in the context of RIS-aided secure communication systems under Fisher-Snedecor $\mathcal{F}$ and can be beneficial in practical applications due to their tightness and simplicity.\vspace{0cm}

\textit{Notations:} $\mathbb{E}(.)$ and $\text{Var}(.)$ are the expectation and variance 
operators respectively, $\Gamma(.)$ and $\gamma\left(.,.\right)$ are the complete and incomplete Gamma functions respectively \cite[Eqs. (8.31) and (8.35)]{zwillinger2007table}, $B(.,.)$ is the Beta function \cite[Eq. (8.38)]{zwillinger2007table}, $G_{p,q}^{m,n}\left(.\right)$ and $G_{p,q:p_1,q_1,p_2,1_2}^{m,n:m_1,n_1:m_2,n_2}\left(.\right)$ are the univariate and bivariate Meijer's G-functions respectively \cite[Eq. (9.3)]{zwillinger2007table}, \cite{gupta1969integrals}, \textcolor{blue}{$H_{p,q}^{m,n}(.)$ is the univariate Fox's H-function \cite[Eq. (1.2)]{mathai2009h}},  $H_{p,q:u,v:e,f}^{m,n:s,t:i,j}\left(.\right)$ is the extended generalized bivariate Fox's H-function \cite[Eq. (2.56)]{mathai2009h}, \textcolor{blue}{$\mathrm{Res}\left[f(x),p\right]$ represents the residue of function $f(x)$ at pole $x=p$}.\vspace{-0.2cm}

\section{System Model}
\subsection{Channel Model}
We consider a secure wireless communication scenario as shown in Fig. \ref{fig-model}, where a legitimate transmitter (Alice) wants to send a confidential message to a legitimate receiver (Bob) aided by a RIS with $N$ reflecting elements, while a passive eavesdropper (Eve) attempts to decode the message sent by Alice. For simplicity, we assume that all nodes are only equipped with single antennas. It is also assumed that the channel state information (CSI) of the legitimate receiver is known to Alice, and thus the RIS can efficiently implement phase shifting in order to maximize the received SNR at Bob. The received signal reflected by the RIS at Bob and the received signal at Eve can be respectively expressed as:
\begin{align}
	Y_\mathrm{B}=\vec{g}_\mathrm{AR}^T\vec{\Phi}\vec{h}_\mathrm{RB}X+Z_\mathrm{B},\label{eq-yb}
\end{align}
\begin{align}
	Y_\mathrm{E}=\vec{g}_\mathrm{AR}^T\vec{\Phi}\vec{k}_\mathrm{RE}X+Z_\mathrm{E},\label{eq-ye}
\end{align}
where $X$ is the transmitted signal, $Z_\mathrm{B}$ and $Z_\mathrm{E}$ denote the additive white Gaussian noise (AWGN) with zero mean and variances $\sigma^2_\mathrm{B}$ and  $\sigma^2_\mathrm{E}$ at Bob and Eve respectively, and $\vec{\Phi}$ is the adjustable phase induced by the $n$th reflecting meta-surface of the RIS which defined as $\vec{\Phi}=\text{diag}\left(\left[\mathrm{e}^{j\phi_1}, \mathrm{e}^{j\phi_2}, ..., \mathrm{e}^{j\phi_N}\right]\right)$.
The vector $\vec{g}^T_\mathrm{AR}$ contains the channel gains from Alice to each element of the RIS, and vector $\vec{h}_\mathrm{RB}$ and $\vec{h}_\mathrm{RE}$ include the channel gains from each element of RIS to Bob and Eve, respectively, which are given by
$
\vec{g}_\mathrm{AR}=d_\mathrm{AR}^{-\alpha}.\left[g_1\mathrm{e}^{-j\psi_1}, g_2\mathrm{e}^{-j\psi_2},..., g_M\mathrm{e}^{-j\psi_N}\right]^T
$,
$
\vec{h}_\mathrm{RB}=d_{RB}^{-\alpha}.\left[h_1\mathrm{e}^{-j\theta_1}, h_2\mathrm{e}^{-j\theta_2},..., h_M\mathrm{e}^{-j\theta_N}\right]^T,
$
and 
$
\vec{h}_\mathrm{RE}=d_{RE}^{-\alpha}.\left[k_1\mathrm{e}^{-j\omega_1}, k_2\mathrm{e}^{-j\omega_2},..., k_M\mathrm{e}^{-j\omega_N}\right]^T,
$
where $d_\mathrm{AR}$ denotes the distance between Alice and the RIS, $d_\mathrm{RB}$ is the distance between the RIS and Bob, $d_\mathrm{RE}$ is the distance between the RIS and Eve, and $\alpha$ defines the path-loss exponent. The terms $g_n$, $h_n$, and $k_n$,   $n\in\{1,2,...,N\}$	are the amplitudes of the corresponding
channel gains, and $\mathrm{e}^{-j\psi_n}$, $\mathrm{e}^{-j\theta_n}$, and $\mathrm{e}^{-j\omega_n}$ denote the phase of
the respective links. We assume that all fading channel coefficients are modeled by the independent Fisher-Snedecor $\mathcal{F}$ distribution with fading parameters $(m_{s_i}, m_i)$ so that $m_{s_i}$ and $m_i$ represent the amount of shadowing of the root-mean-square (rms) signal power and the fading severity parameter, respectively.
\begin{figure}[!t]\vspace{0ex}
	\centering
	\includegraphics[width=0.45\textwidth]{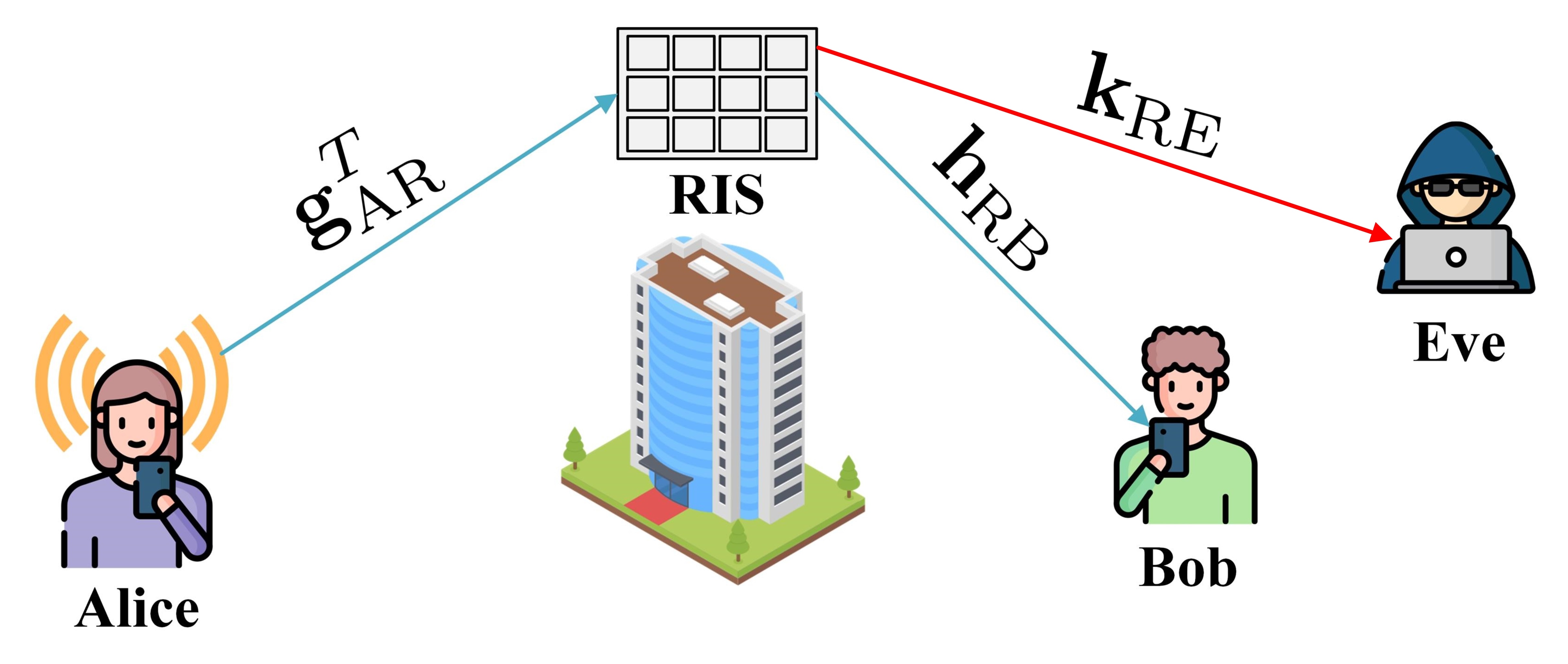} 
	\caption{System model depicting the RIS-aided secure communications.} 
	\label{fig-model}
\end{figure}

\subsection{SNR Distribution}
\subsubsection{The main link} From \eqref{eq-yb}, the instantaneous SNR at Bob can be expressed as:\vspace{0cm}
\begin{align}
	\gamma_\mathrm{B}&=\frac{P\Big|\sum_{n=1}^{N}g_nh_n\mathrm{e}^{j\left(\phi_n-\psi_n-\theta_n\right)}\Big|^2}{\sigma_\mathrm{B}^2d_\mathrm{AR}^\alpha d_\mathrm{RB}^\alpha},\\
	&\overset{(a)}{=}\frac{P\Big|\sum_{n=1}^{N}g_nh_n\Big|^2}{\sigma_\mathrm{B}^2d_\mathrm{AR}^\alpha d_\mathrm{RB}^\alpha}=\bar{\gamma}_\mathrm{B}W,
\end{align}
where $\bar{\gamma}_\mathrm{B}$ is the average SNR at Bob and $(a)$ is obtained by considering perfect knowledge of the CSI at RIS, which provides the ideal phase shifting (i.e., $\phi_n=\psi_n+\theta_n$) for maximizing the SNR at Bob \cite{yang2020secrecy}. In the following theorem, we present the distributions of $\gamma_{\mathrm{B}}$, assuming that $g_n$ and $h_n$ follow independent Fisher-Snedecor $\mathcal{F}$ distribution.

\begin{theorem}\label{thm-pdf-cdf}
	Assuming all channels follow  Fisher-Snedecor $\mathcal{F}$ fading model, the PDF and the CDF of $\gamma_{\mathrm{B}}$ are respectively determined as:
	\begin{align}
		f_{\gamma_\mathrm{B}}(\gamma_\mathrm{B})=\frac{\gamma_\mathrm{B}^{(a-1)/2}}{2\bar{\gamma}_\mathrm{B}^{(a+1)/2}b^{a+1}\Gamma(a+1)}\mathrm{e}^{-\frac{\sqrt{\gamma_\mathrm{B}}}{b\sqrt{\bar{\gamma}_\mathrm{B}}}},\label{eq-pdf}
	\end{align}
	\begin{align}
		F_{\gamma_\mathrm{B}}(\gamma_\mathrm{B})=\frac{\gamma\left(a+1,\frac{\sqrt{\gamma_\mathrm{B}}}{b\sqrt{\bar{\gamma}_\mathrm{B}}}\right)}{\Gamma(a+1)},\label{eq-cdf}
	\end{align}
	where $\Omega_i$ for $i\in\left\{\mathrm{R},\mathrm{B}\right\}$ is the mean power, 
	\begin{align}\nonumber
		&	a=\frac{\left(N+1\right)\mathcal{B}^2-\mathcal{AC}}{\mathcal{AC}-\mathcal{B}^2},\quad b=\frac{\mathcal{D}\left(\mathcal{AC}-\mathcal{B}^2\right)}{\mathcal{BC}},\\\nonumber
		&\mathcal{A}=B\left(m_\mathrm{B}+1,m_{s_\mathrm{B}}-1\right)B\left(m_\mathrm{R}+1,m_{s_\mathrm{R}}-1\right),\\\nonumber
		&	\mathcal{B}=B\left(m_\mathrm{B}+\frac{1}{2},m_{s_\mathrm{B}}-\frac{1}{2}\right)B\left(m_\mathrm{R}+\frac{1}{2},m_{s_\mathrm{R}}-\frac{1}{2}\right)\\\nonumber 
		&\mathcal{C}=B\left(m_\mathrm{B},m_{s_\mathrm{B}}\right)B\left(m_\mathrm{R},m_{s_\mathrm{R}}\right), \text{and} \,\\ \nonumber &\mathcal{D}=\left(\frac{m_\mathrm{B}m_\mathrm{R}}{(m_{s_\mathrm{B}}-1)(m_{s_\mathrm{R}}-1)\Omega_\mathrm{B}\Omega_\mathrm{R}}\right)^{-\frac{1}{2}}.
	\end{align}
\end{theorem}
\begin{proof}
	The details of proof are in Appendix \ref{app-pdf}. 
\end{proof}
\subsubsection{The eavesdropping link}
For the eavesdropper link, Eve receives signals from the RIS. Hence, from \eqref{eq-ye}, the instantaneous SNR at Eve can be expressed as:
\begin{align}
	\gamma_\mathrm{E}&=\frac{P\Big|\sum_{n=1}^{N}g_nk_n\mathrm{e}^{j\left(\phi_n-\psi_n-\omega_n\right)}\Big|^2}{\sigma_\mathrm{E}^2d_\mathrm{AR}^\alpha d_\mathrm{RE}^\alpha},\\
	&\overset{(a)}{=}\frac{P\Big|\sum_{n=1}^{N}g_nk_n\Big|^2}{\sigma_\mathrm{B}^2d_\mathrm{AR}^\alpha d_\mathrm{RE}^\alpha}=\bar{\gamma}_\mathrm{E}Q,
\end{align}
where $\bar{\gamma}_\mathrm{E}$ is the average SNR. Similarly, assuming that $g_n$ and $k_n$ follow independent Fisher-Snedecor $\mathcal{F}$ distribution, the distributions of $\gamma_\mathrm{E}$ can be determined as the following corollary. 
\begin{corollary}
	Assuming all channels follow  Fisher-Snedecor $\mathcal{F}$ fading model, the PDF and the CDF of $\gamma_{\mathrm{E}}$ are respectively determined as:
	\begin{align}
		f_{\gamma_\mathrm{E}}(\gamma_\mathrm{E})=\frac{\gamma_\mathrm{E}^{(a'-1)/2}}{2\bar{\gamma}_\mathrm{E}^{(a'+1)/2}b'^{a+1}\Gamma(a'+1)}\mathrm{e}^{-\frac{\sqrt{\gamma_\mathrm{E}}}{b'\sqrt{\bar{\gamma}_\mathrm{E}}}},\label{eq-pdf2}
	\end{align}
	\begin{align}
		F_{\gamma_\mathrm{E}}(\gamma_\mathrm{E})=\frac{\gamma\left(a'+1,\frac{\sqrt{\gamma_\mathrm{E}}}{b'\sqrt{\bar{\gamma}_\mathrm{E}}}\right)}{\Gamma(a'+1)},\label{eq-cdf2}
	\end{align}
	where $\Omega_j$ for $j\in\left\{\mathrm{R},\mathrm{B}\right\}$ is the mean power, 
	\begin{align}\nonumber
		&	a'=\frac{\left(N+1\right)\mathcal{B'}^2-\mathcal{A'C'}}{\mathcal{A'C'}-\mathcal{B'}^2},\quad b'=\frac{\mathcal{D'}\left(\mathcal{A'C'}-\mathcal{B'}^2\right)}{\mathcal{B'C'}},\\\nonumber &\mathcal{A'}=B\left(m_\mathrm{E}+1,m_{s_\mathrm{E}}-1\right)B\left(m_\mathrm{R}+1,m_{s_\mathrm{R}}-1\right),\\\nonumber
		&	\mathcal{B'}=B\left(m_\mathrm{E}+\frac{1}{2},m_{s_\mathrm{E}}-\frac{1}{2}\right)B\left(m_\mathrm{R}+\frac{1}{2},m_{s_\mathrm{R}}-\frac{1}{2}\right),\\\nonumber
		&\mathcal{C'}=B\left(m_\mathrm{E},m_{s_\mathrm{E}}\right)B\left(m_\mathrm{R},m_{s_\mathrm{R}}\right), \text{and}\\\nonumber &\mathcal{D'}=\left(\frac{m_\mathrm{E}m_\mathrm{R}}{(m_{s_\mathrm{E}}-1)(m_{s_\mathrm{R}}-1)\Omega_\mathrm{E}\Omega_\mathrm{R}}\right)^{-\frac{1}{2}}.
	\end{align}
\end{corollary}
\begin{proof}
	The proof can directly obtained form Appendix \ref{app-pdf}.
\end{proof}

\section{Secrecy Performance Analysis}
In this section, we derive closed-form expressions of the SOP and the ASC for the considered system model. \vspace{0cm}
\subsection{SOP analysis}
The SOP is an information-theoretical metric to evaluate the performance of physical layer security, which is defined as the probability that the instantaneous secrecy capacity $C_\mathrm{s}$ is less than a target secrecy rate $R_\mathrm{s}>0$, $P_{\mathrm{sop}}=\Pr\left(C_\mathrm{s}\leq R_\mathrm{s}\right)$,  
where $C_\mathrm{s}$ for the considered system model can be defined as:
\begin{align}
	C_\mathrm{s}\left(\gamma_{\mathrm{B}},\gamma_{\mathrm{E}}\right)=\max\left\{\ln\left(1+\gamma_{\mathrm{B}}\right)-\ln\left(1+\gamma_{\mathrm{E}}\right),0\right\}.\label{eq-sc}
\end{align}
Now, by inserting \eqref{eq-sc} into SOP definition, $P_{\mathrm{sop}}$ for the considered model is mathematically expressed as:
\begin{align}
	P_{\mathrm{sop}}&=\Pr\left(\ln\left(\frac{1+\gamma_\mathrm{B}}{1+\gamma_\mathrm{E}}\right)\leq R_\mathrm{s}\right),\\
	&=\int_0^\infty F_{\gamma_\mathrm{B}}\left(\gamma_\mathrm{t}\right)f_{\gamma_\mathrm{E}}(\gamma_\mathrm{E})d\gamma_\mathrm{E},\label{eq-sop1}
\end{align}
where $\gamma_\mathrm{t}=\left(1+\gamma_\mathrm{E}\right)\mathrm{e}^{R_\mathrm{s}}-1=\gamma_\mathrm{E}R_\mathrm{t}+R_\mathrm{t}-1=\gamma_\mathrm{E}R_\mathrm{t}+\bar{R}_\mathrm{t}$.
\begin{theorem}\label{thm-sop}
	The SOP for the considered RIS-aided secure communication system under Fisher-Snedecor $\mathcal{F}$ fading channels is given by 
	\begin{align}
		\hspace{-0.2cm}	P_\mathrm{sop}=\mathcal{E}H_{1,0:2,2:1,1}^{0,1:1,1:1,1}\left(
		\begin{array}{c}
			\frac{b\sqrt{\bar{\gamma}_\mathrm{B}}}{\sqrt{\bar{R}_\mathrm{t}}},\frac{b'^2R_\mathrm{t}\bar{\gamma}_\mathrm{E}}{\bar{R}_\mathrm{t}}\\
		\end{array}\Bigg\vert
		\begin{array}{c}
			\mathcal{I}_1\\
			-\\
		\end{array}\Bigg\vert
		\begin{array}{c}
			\mathcal{I}_2\\
			\mathcal{J}_2
		\end{array}\Bigg\vert
		\begin{array}{c}
			\mathcal{I}_3\\
			\mathcal{J}_3
		\end{array}
		\right), \label{eq-sop}
	\end{align}
	where \textcolor{blue}{$\mathcal{E}=\frac{\bar{R}_\mathrm{t}b'^{-2}}{R_\mathrm{t}\bar{\gamma}_\mathrm{E}\Gamma\left(a'+1\right)}$}, $\mathcal{I}_1=\left(2;\frac{1}{2},1\right)$, $\mathcal{I}_2=\left(-a,1\right),\left(1,1\right)$, $\mathcal{I}_3=\left(2-a',2\right)$, $\mathcal{J}_2=	\left(0,1\right), \left(1,\frac{1}{2}\right)$, and $\mathcal{J}_3=\left(1,1\right)$.
	%
\end{theorem}
\begin{proof}
	The details of proof are in Appendix \ref{app1}.
\end{proof}
\subsection{ASC analysis}
Here, we define the ASC for the system model under consideration, as:
\begin{align}
	\bar{C}_\mathrm{s}\overset{\Delta}{=}\int_{0}^\infty\int_{0}^\infty C_\mathrm{s}\left(\gamma_\mathrm{B},\gamma_\mathrm{E}\right)f_{\gamma_\mathrm{B}}\left(\gamma_\mathrm{B}\right)f_{\gamma_\mathrm{E}}\left(\gamma_\mathrm{E}\right)d\gamma_\mathrm{B}d\gamma_\mathrm{E}.\label{eq-def-sc}
\end{align}
\begin{theorem}\label{thm-asc}
	The ASC for the considered RIS-aided secure communication system under Fisher-Snedecor $\mathcal{F}$ fading channels is given by 
	\begin{align}\nonumber
		&\bar{C}_\mathrm{s}=\mathcal{H}H_{1,0:1,2:2,2}^{0,1:1,1:1,2}\left(
		\begin{array}{c}
			\frac{b\sqrt{\bar{\gamma}_\mathrm{B}}}{b'\sqrt{\bar{\gamma}_\mathrm{E}}},b^2\bar{\gamma}_\mathrm{B}\\
		\end{array}\Bigg\vert
		\begin{array}{c}
			\mathcal{U}_1\\
			-\\
		\end{array}\Bigg\vert
		\begin{array}{c}
			\mathcal{U}_2\\
			\mathcal{V}_2
		\end{array}\Bigg\vert
		\begin{array}{c}
			\mathcal{U}_3\\
			\mathcal{V}_3
		\end{array}
		\right)\\\nonumber
		&+\mathcal{H}H_{1,0:1,2:2,2}^{0,1:1,1:1,2}\left(
		\begin{array}{c}
			\frac{b'\sqrt{\bar{\gamma}_\mathrm{E}}}{b\sqrt{\bar{\gamma}_\mathrm{B}}},b^2\bar{\gamma}_\mathrm{E}\\
		\end{array}\Bigg\vert
		\begin{array}{c}
			\mathcal{W}_1\\
			-\\
		\end{array}\Bigg\vert
		\begin{array}{c}
			\mathcal{W}_2\\
			\mathcal{X}_2
		\end{array}\Bigg\vert
		\begin{array}{c}
			\mathcal{W}_3\\
			\mathcal{X}_3
		\end{array}
		\right)\\
		&-\mathcal{O} G_{3,2}^{1,3}\left(
		\begin{array}{c}
			4b'^2\bar{\gamma}_\mathrm{E}\\
		\end{array}\Bigg\vert
		\begin{array}{c}
			(-\frac{a'}{2},1,1)\\
			(1,0)\\
		\end{array}\right),\label{eq-asc}
	\end{align}
	where $\mathcal{H}=\frac{1}{\Gamma(a+1)\Gamma(a'+1)\ln 2}$, $\mathcal{O}=\frac{2^{\frac{a'-2}{2}}}{\sqrt{2\pi}\Gamma\left(a'+1\right)\ln 2}$, $\mathcal{U}_1=		\left(-a;1,2\right)$, $\mathcal{U}_2=\mathcal{W}_2=\left(1,1\right)$, $\mathcal{U}_3=\mathcal{W}_3=\left(1,1\right),\left(1,1\right)$, $\mathcal{V}_2=	\left(a'+1,1\right),\left(0,1\right)$,  $\mathcal{V}_3=\mathcal{X}_3=		\left(1,1\right), \left(0,1\right)$, $\mathcal{W}_1=\left(-a';1,2\right)$, and  $\mathcal{X}_2=	\left(a+1,1\right),\left(0,1\right)$.
\end{theorem}
\begin{proof}
	The details of proof are in Appendix \ref{app2}.
\end{proof}
\textcolor{blue}{\begin{remark}
		The SOP and the ASC of the considered RIS-aided communication system can be accurately obtained according to Theorems \ref{thm-sop} and \ref{thm-asc}, respectively. We can see that the SOP and ASC highly depend on the number of RIS elements $N$ and the fading parameters $(m_{s_i},m_i)$ so that the performance of SOP and ASC will improve when $N$ and $m_i$ increase because of better communication conditions.
\end{remark}}
\textcolor{blue}{\section{Asymptotic Analysis of Secrecy Metrics}
	In order to gain more insight into the behaviour of secrecy metrics in the high SNR regime, we analyze the asymptotic behaviour of both SOP and ASC in this section.
	\subsection{Asymptotic SOP}
	Here, in order to derive the asymptotic expression of the SOP, we exploit the expansions of the univariate and bivariate Fox's H-functions, which can be obtained by evaluating the residue of the corresponding integrands at the closest poles to the contour, namely, the minimum pole on the right for large Fox's H-function arguments and the maximum pole on the left for small ones \cite{chergui2016performance}. Regarding the exact expression that obtained for the SOP in \eqref{eq-sop}, at high SNR regime, i.e., $\bar{\gamma}_\mathrm{B}\rightarrow{\infty}$, we have $\frac{b\sqrt{\Bar{\gamma}_\mathrm{B}}}{\sqrt{\bar{R}_\mathrm{t}}}\rightarrow{\infty}$. In this case, the asymptotic SOP can be determined according to the following proposition.
	\begin{proposition}\label{pro-sop}
		The asymptotic SOP (i.e., $\bar{\gamma}_\mathrm{B}\rightarrow{\infty}$) for the considered RIS-aided secure communication system under Fisher-Snedecor $\mathcal{F}$ fading channels is determined as:
		\begin{align}
			P_{\mathrm{sop}}^{\mathrm{asy}}=\frac{\Gamma(a+a'+2)\Gamma(a+1)}{\Gamma(a'+1)\Gamma(a+2)}\left(\frac{b'^2R_\mathrm{t}\bar{\gamma}_\mathrm{E}}{b^2\bar{\gamma}_\mathrm{B}}\right)^{\frac{a+1}{2}}.\label{eq-asy-sop}
		\end{align}
	\end{proposition}
	\begin{proof}
		The details of proof are in Appendix \ref{app3}.
	\end{proof}
	\subsection{Asymptotic ASC}
	By applying the same approach that we exploited for finding the asymptotic SOP, we determine the asymptotic expression for the ASC at the high SNR regime (i.e., $\bar{\gamma}_\mathrm{B}\rightarrow \infty$) in the following proposition.
	\begin{proposition}\label{pro-asc}
		The asymptotic ASC (i.e., $\bar{\gamma}_\mathrm{B}\rightarrow{\infty}$) for the considered RIS-aided secure communication system under Fisher-Snedecor $\mathcal{F}$ fading channels is determined as:
		\begin{align}\nonumber
			&\bar{C}^{\mathrm{asy}}_\mathrm{s}=\mathcal{H}_1H_{2,3}^{4,3}\left(
			\begin{array}{c}
				\hspace{-0.3cm}	b'\sqrt{\bar{\gamma}_\mathrm{E}}\\
			\end{array}\Bigg\vert
			\begin{array}{c}
				\left(\frac{a+3}{2},\frac{1}{2}\right),\left(\frac{a+3}{2},\frac{1}{2}\right),\left(-a',1\right)\\
				\left(\frac{a+3}{2},\frac{1}{2}\right),(0,1),\left(\frac{a+1}{2},\frac{1}{2}\right)\\
			\end{array}\hspace{-1ex}\right)\\\nonumber
			&+\mathcal{H}H_{2,1}^{2,2}\left(
			\begin{array}{c}
				\hspace{-0.2cm}	\frac{b'\sqrt{\bar{\gamma}_\mathrm{E}}}{b\sqrt{\bar{\gamma}_\mathrm{B}}}\\
			\end{array}\Bigg\vert
			\begin{array}{c}
				\left(-a',1\right),\left(1,1\right)\\
				\left(0,1\right),(1+a,1)
			\end{array}\right)\\\nonumber
			&+\mathcal{H}_2H_{2,3}^{4,3}\left(
			\begin{array}{c}
				\hspace{-0.2cm}	b\sqrt{\bar{\gamma}_\mathrm{B}}\\
			\end{array}\Bigg\vert
			\begin{array}{c}
				\left(\frac{a'+3}{2},\frac{1}{2}\right),\left(\frac{a'+3}{2},\frac{1}{2}\right),\left(-a,1\right)\\
				\left(\frac{a'+3}{2},\frac{1}{2}\right),(0,1),\left(\frac{a'+1}{2},\frac{1}{2}\right)\\
			\end{array}\hspace{-1ex}\right)\\\nonumber
			&+\mathcal{H}H_{2,1}^{2,2}\left(
			\begin{array}{c}
				\hspace{-0.2cm}	\frac{b\sqrt{\bar{\gamma}_\mathrm{B}}}{b'\sqrt{\bar{\gamma}_\mathrm{E}}}\\
			\end{array}\Bigg\vert
			\begin{array}{c}
				\left(-a,1\right),\left(1,1\right)\\
				\left(0,1\right),(1+a',1)
			\end{array}\right)\\
			&-\mathcal{O} G_{3,2}^{1,3}\left(
			\begin{array}{c}
				4b'^2\bar{\gamma}_\mathrm{E}\\
			\end{array}\Bigg\vert
			\begin{array}{c}
				(-\frac{a'}{2},1,1)\\
				(1,0)\\
			\end{array}\right),\label{eq-asy-asc}
		\end{align}
		where $\mathcal{H}_1=\mathcal{H}\left(b^2\bar{\gamma}_\mathrm{B}\right)^{-\frac{1+a}{2}}$ and $\mathcal{H}_2=\mathcal{H}\left(b'^2\bar{\gamma}_\mathrm{E}\right)^{-\frac{1+a'}{2}}$.
	\end{proposition}
	\begin{remark}
		By comparing Theorems \ref{thm-sop} and \ref{thm-asc} with Propositions \ref{pro-sop} and \ref{pro-asc}, we can see that the SOP and ASC which are derived in terms of the bivariate Fox's H-function can be simplified to the univariate Fox's H-function at the high SNR regime, by exploiting the residue approach. We can also observe that the asymptotic results in \eqref{eq-asy-sop} and \eqref{eq-asy-asc} will gradually approximate the exact SOP and ASC, respectively, with high accuracy as $\bar{\gamma}_\mathrm{B}\rightarrow\infty$.
\end{remark}}
\section{Numerical Results}
In this section, we evaluate the theoretical expressions previously derived, which are double-checked in all instances with Monte-Carlo (MC) simulations. Besides, it should be noted that although the implementation
of the extended generalized bivariate Fox's H-function is not available in mathematical packages, like Mathematica,
Maple, or MATLAB, it is computationally tractable and programmable as explained in \cite{peppas2012new}. \textcolor{blue}{We also set $P=30$dBm, $\sigma^2_{\mathrm{B}}=-40$dBm, $\sigma^2_{\mathrm{B}}=-20$dBm, $d_\mathrm{AR}=d_\mathrm{RB}=d_\mathrm{RE}=10$m, and $\alpha=3.$}

Fig. \ref{fig-c-gamma} illustrates the behavior of the SOP in terms of $\bar{\gamma}_\mathrm{B}$ for the selected value of $N$ under Fisher-Snedecor $\mathcal{F}$ fading channels. It can be seen as the number of RIS reflecting elements increases the SOP performance improves. The reason is that exploiting the RIS can provide a channel with higher quality and also enhance the received SNR when phase processing is utilized. Fig. \ref{fig-out-m} shows the SOP performance in terms of $\bar{\gamma}_\mathrm{B}$ for different  values of fading parameters $m_i$, $i\in\left\{\mathrm{R},\mathrm{B},\mathrm{E}\right\}$ under RIS deployment. We can see that as the fading is less severe (i.e., as $m_i$ increases) the SOP performance improves, meaning that system performance ameliorates (degrades) in environments that exhibit light (heavy) fading characteristics. Figs. \ref{fig-c-gamma} and \ref{fig-c-m} represent the ASC performance based on the variation of $\bar{\gamma}_{\mathrm{B}}$ for selected values of $N$ and $\bar{\gamma}_{\mathrm{E}}$, respectively. As shown in Fig. \ref{fig-c-gamma}, a relatively small increase in the number of RIS reflecting elements has a constructive impact on the ASC performance, regardless of the main and eavesdropper channels condition. Fig. \ref{fig-c-m} shows that under RIS deployment and Fisher-Snedecor $\mathcal{F}$ fading channels, the ASC performance improves as $\bar{\gamma}_\mathrm{B}$ grows. Moreover, we can see that the ASC can still be achieved even if $\bar{\gamma}_\mathrm{B}\leq\bar{\gamma}_\mathrm{E}$. In summary, both analytical and simulation results show that considering the RIS and then increasing the RIS’s elements offer more degrees of freedom
for efficient beamforming which
can significantly improve the system performance. 

\begin{figure}[!t]\vspace{0ex}
	\centering
	\includegraphics[width=0.5\textwidth]{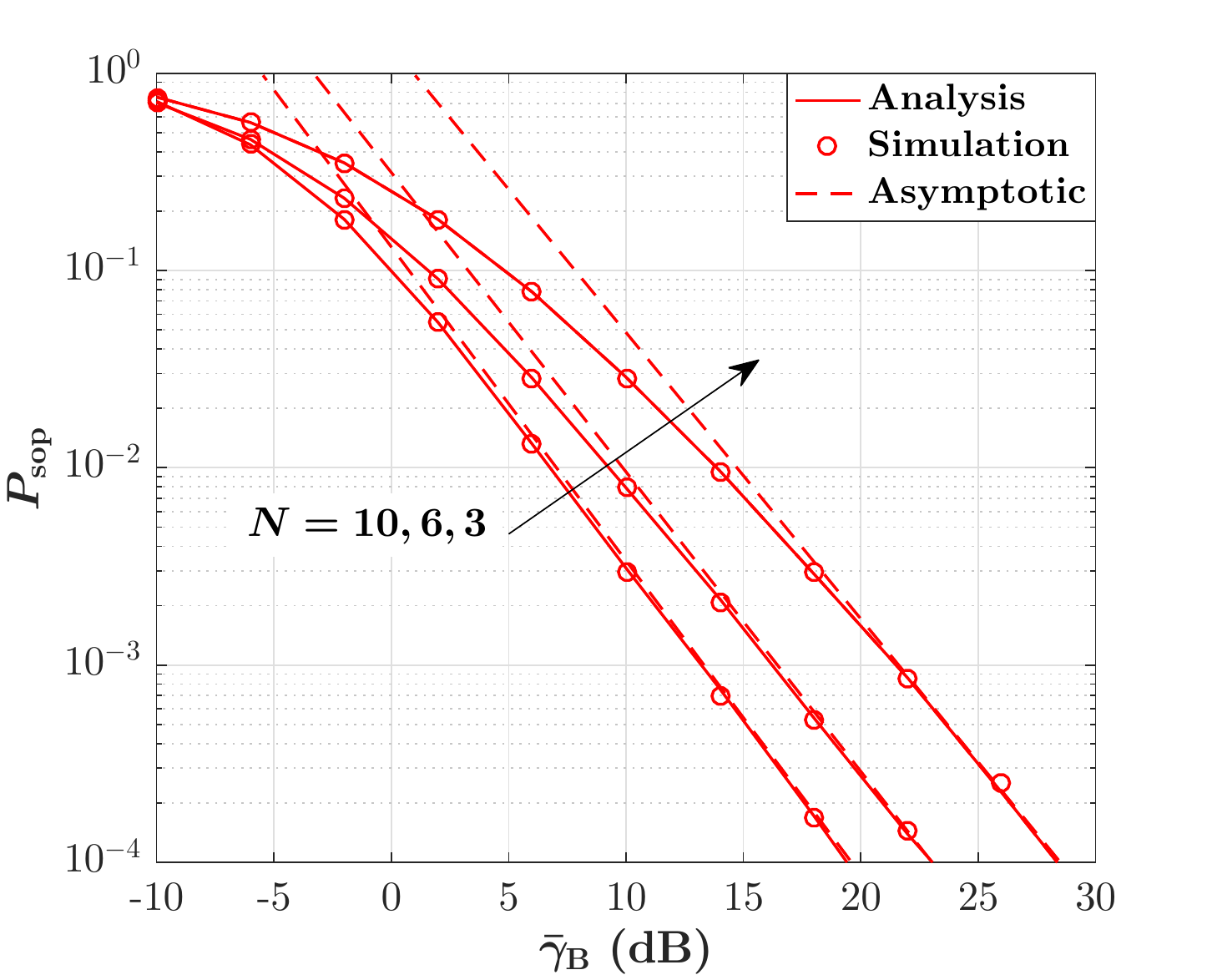} 
	\caption{SOP versus $\bar{\gamma}_\mathrm{B}$ for selected values of $N$, when $\bar{\gamma}_\mathrm{E}=-10$dB, $R_\mathrm{s}=1$ and $\left(m_{s_i},m_i\right)=\left(3,2\right)$.} 
	\label{fig-out-gamma}
\end{figure}
\begin{figure}[!t]\vspace{0ex}
	\centering
	\includegraphics[width=0.5\textwidth]{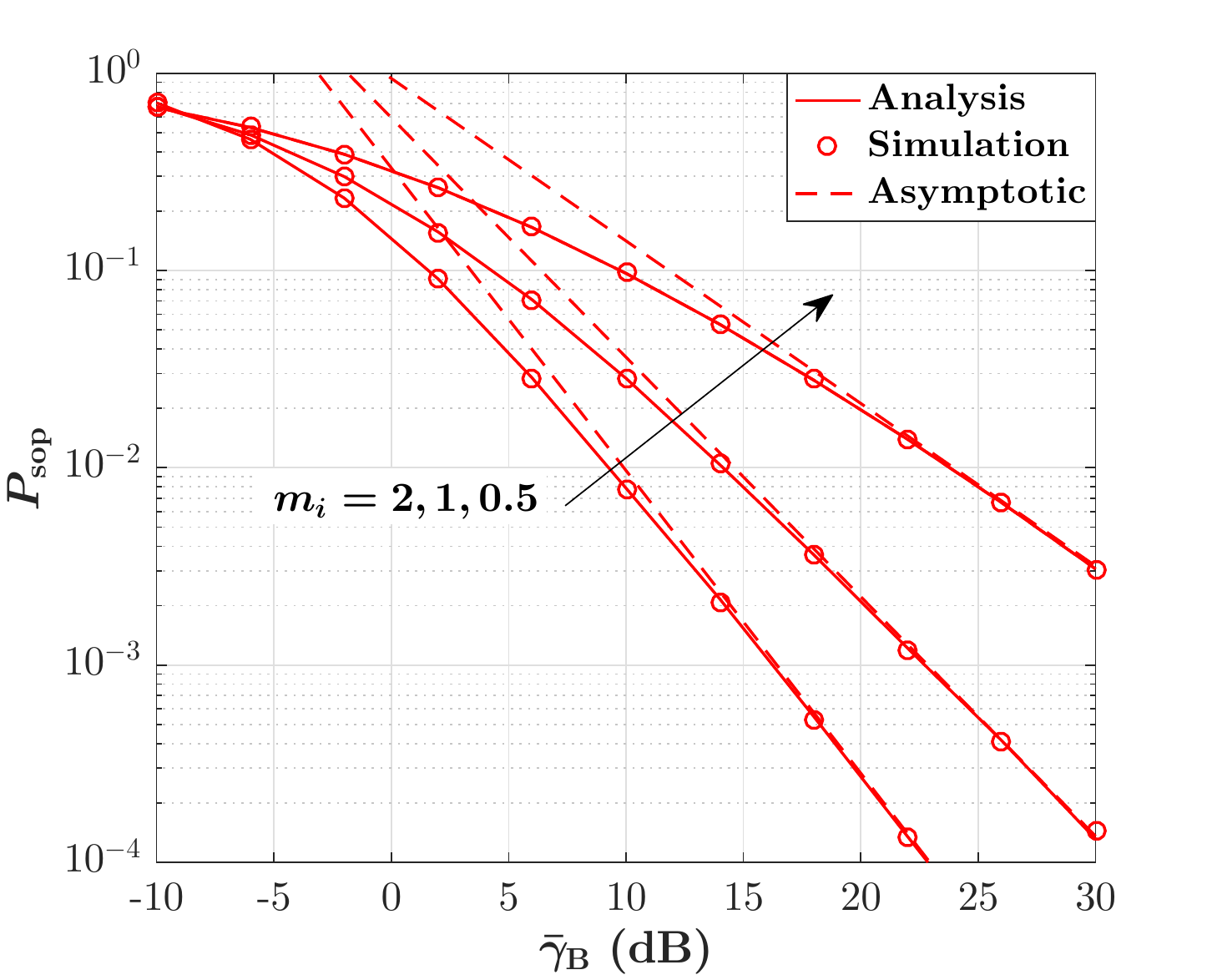} 
	\caption{SOP versus $\bar{\gamma}_\mathrm{B}$ for selected values of $m_i$, when $N=6$, $m_{s_i}=3$, and $R_\mathrm{s}=1$.} 
	\label{fig-out-m}
\end{figure}
\begin{figure}[!h]\vspace{0ex}
	\centering
	\includegraphics[width=0.5\textwidth]{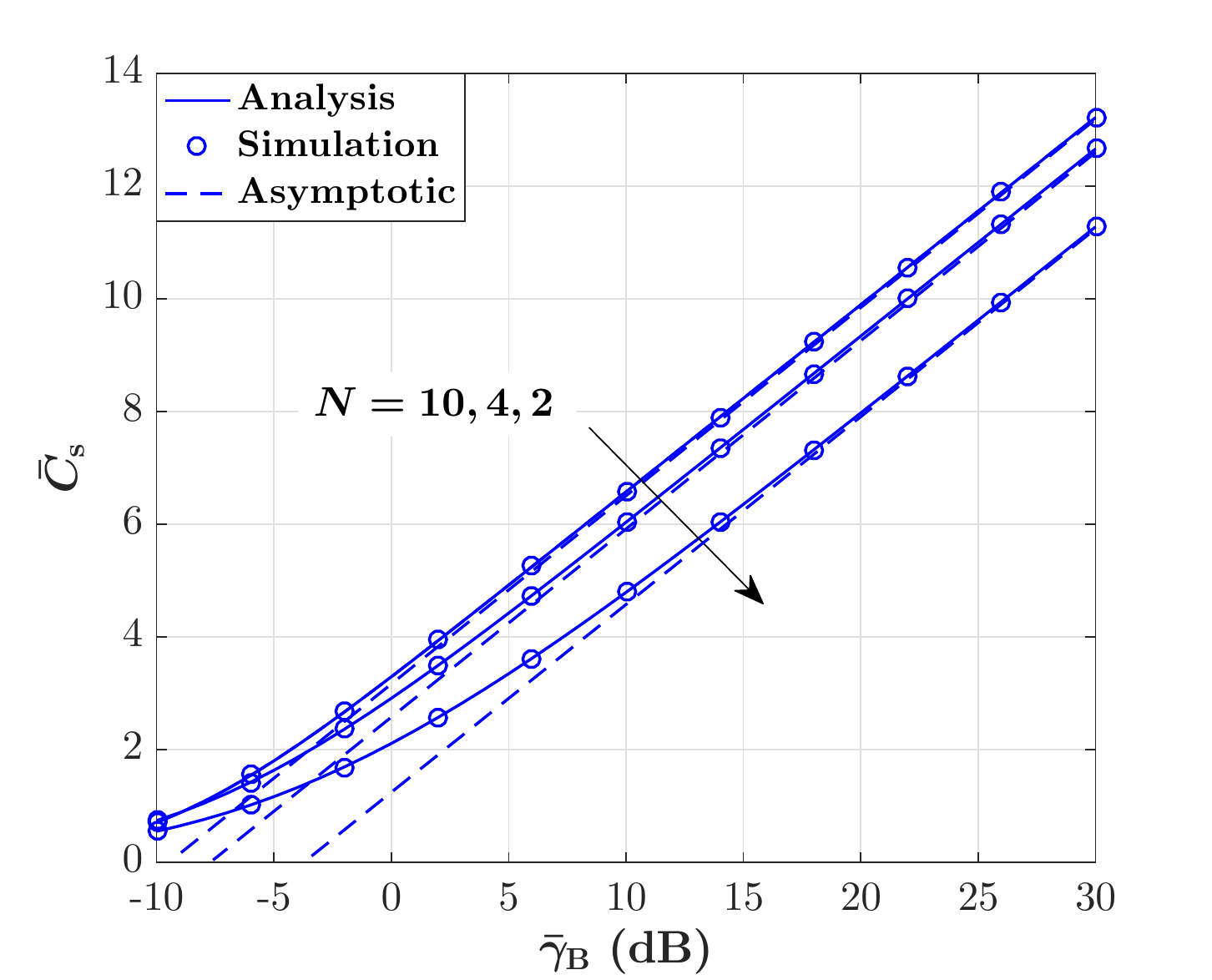} 
	\caption{ASC versus $\bar{\gamma}_\mathrm{B}$ for selected values of $N$, when $\bar{\gamma}_\mathrm{E}=-10$dB, $R_\mathrm{s}=1$, and $\left(m_{s_i},m_i\right)=\left(3,2\right)$.} 
	\label{fig-c-gamma}
\end{figure}
\begin{figure}[!h]\vspace{0ex}
	\centering
	\includegraphics[width=0.5\textwidth]{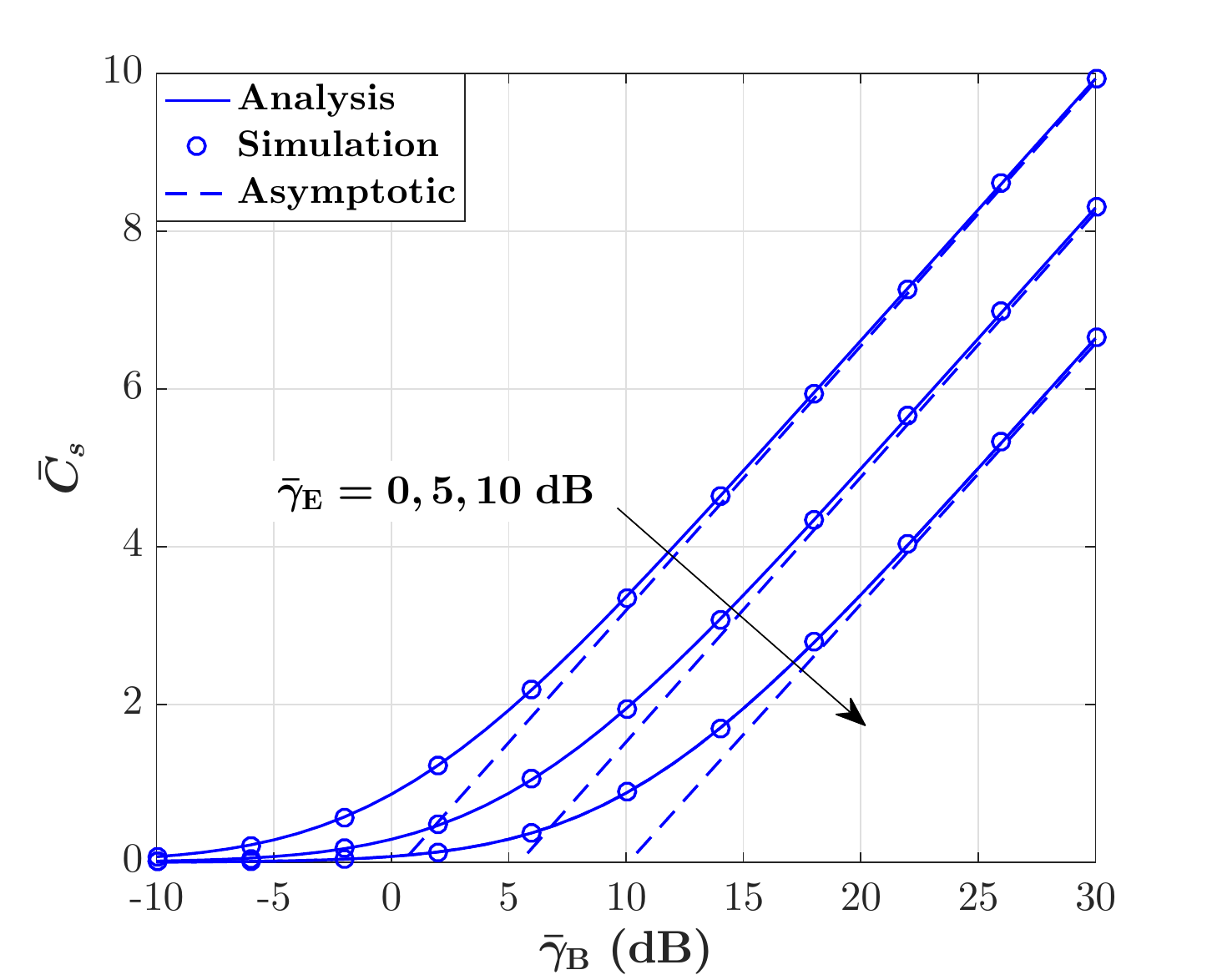} 
	\caption{ASC versus $\bar{\gamma}_\mathrm{B}$ for selected values of $\bar{\gamma}_\mathrm{E}$, when $N=6$, $R_\mathrm{s}=1$, and $\left(m_{s_i},m_i\right)=\left(3,2\right)$.} 
	\label{fig-c-m}
\end{figure}

\section{Conclusion}\vspace{-0.1cm}
In this paper, we studied the secrecy performance of a RIS-aided wireless communication system under Fisher-Snedecor $\mathcal{F}$ fading channels. In particular, we first presented the analytical expressions for the PDF and the CDF of corresponding SNRs, and then we derived the closed-form expressions of the SOP and the ASC. \textcolor{blue}{We also analyzed the asymptotic behaviour of the obtained secrecy metrics by using the residue method.} The simulation results validated the accuracy of analytical results and showed that the considering RIS between the transmitter and other nodes has constructive effects on system performance in terms of the SOP and ASC. 
\appendix
\subsection{Proof of Theorem 1}\label{app-pdf}
Let $W_n=G_nH_n$ be the product of two independent Fisher-Snedecor $\mathcal{F}$ RVs. Then, the PDF of $W_n$ is defined as \cite{badarneh2020product}:
\begin{align}\nonumber
	&f_{W_n}(w_n)=\frac{2w_n^{-1}}{\Gamma(m_\mathrm{R})\Gamma(m_{s_\mathrm{R}})\Gamma(m_\mathrm{B})\Gamma(m_{s_\mathrm{B}})} \\
	&\times G_{2,2}^{2,2}\left(
	\begin{array}{c}
		\frac{w_n^2 m_\mathrm{R}m_\mathrm{B}}{(m_{s_\mathrm{R}}-1)(m_{s_\mathrm{B}}-1)\omega_2\omega_1}\\
	\end{array}\Bigg\vert
	\begin{array}{c}
		1-m_{s_\mathrm{R}},1-m_{s_\mathrm{B}}\\
		m_\mathrm{R},m_\mathrm{B}\\
	\end{array}\right).
\end{align}
The mean and variance of $W_i$ are respectively given by \cite{badarneh2020product}
\begin{align}\nonumber
	&\mathbb{E}\left(W_n\right)=\frac{B\left(m_\mathrm{B}+\frac{1}{2},m_{s_\mathrm{B}}-\frac{1}{2}\right)B\left(m_\mathrm{R}+\frac{1}{2},m_{s_\mathrm{R}}-\frac{1}{2}\right)}{B\left(m_\mathrm{B},m_{s_\mathrm{B}}\right)B\left(m_\mathrm{R},m_{s_\mathrm{R}}\right)}\\
	&\times\left(\frac{m_\mathrm{B}m_\mathrm{R}}{(m_{s_\mathrm{B}}-1)(m_{s_\mathrm{R}}-1)\Omega_\mathrm{B}\Omega_\mathrm{R}}\right)^{-\frac{1}{2}},
\end{align}
\begin{align}\nonumber
	&	\text{Var}\left(W_n\right)=\frac{(m_{s_\mathrm{B}}-1)(m_{s_\mathrm{R}}-1)\Omega_\mathrm{B}\Omega_\mathrm{R}}{m_\mathrm{B}m_\mathrm{R}B\left(m_\mathrm{B},m_{s_\mathrm{B}}\right)B\left(m_\mathrm{R},m_{s_\mathrm{R}}\right)}\\\nonumber
	&\times\Bigg[\frac{B\left(m_\mathrm{B}+\frac{1}{2},m_{s_\mathrm{B}}-\frac{1}{2}\right)^2B\left(m_\mathrm{R}+\frac{1}{2},m_{s_\mathrm{R}}-\frac{1}{2}\right)^2}{B\left(m_\mathrm{B},m_{s_\mathrm{B}}\right)B\left(m_\mathrm{R},m_{s_\mathrm{R}}\right)}\\
	&-B\left(m_\mathrm{B}+1,m_{s_\mathrm{B}}-1\right)B\left(m_\mathrm{R}+1,m_{s_\mathrm{R}}-1\right)\Bigg].
\end{align}
Now, by exploiting the first term of a Laguerre expansion \cite[Sec. 2.2.2]{primak2005stochastic}, the PDF of $W=\sum_{n=1}^NW_n$ can be accurately determined as:
\begin{align}
	f_W(w)=\frac{w^a}{b^{a+1}\Gamma(a+1)}\mathrm{e}^{-\frac{w}{b}},
\end{align}
where $a=\frac{N\mathbb{E}(W_n)^2}{\text{Var}(W_n)}-1$ and  $b=\frac{\text{Var}(W_n)}{\mathbb{E}(W_n)}$. Eventually, by performing the RV transformation $f_{\gamma_\mathrm{B}}(\gamma_\mathrm{B})=\frac{1}{2\sqrt{\bar{\gamma}_\mathrm{B}\gamma_\mathrm{B}}}f_W\left(\sqrt{\frac{\gamma_\mathrm{B}}{\bar{\gamma}_\mathrm{B}}}\right)$, the PDF of $\gamma_{\mathrm{B}}$ is obtained as \eqref{eq-pdf}. Besides, by using definition $F_{\gamma_\mathrm{B}}(\gamma_\mathrm{B})=\int_{0}^\gamma f_{\gamma_\mathrm{B}}(x)dx$, and then computing the corresponding integral, 
the CDF of $\gamma_{\mathrm{B}}$ is also achieved as \eqref{eq-cdf}. 
\subsection{Proof of  Theorem 2}\label{app1}
By exploiting the Parseval relation for Mellin transform \cite[Eq. (8.3.23)]{davies2002integral}, \eqref{eq-sop1} can be rewritten as follows:
\begin{align}
	\hspace{-0.3cm}	P_\mathrm{sop}=\frac{1}{2\pi j}\int_{\mathcal{L}_1}\mathcal{M}\left[F_{\gamma_\mathrm{B}}\left(\gamma_\mathrm{t}\right),1-s\right]\mathcal{M}\left[f_{\gamma_\mathrm{E}}\left(\gamma_\mathrm{E}\right),s\right]ds,\label{eq-sop11}
\end{align}
where $\mathcal{L}_1$ is the integration path from $\nu-j\infty$ to $\nu+j\infty$ for
a constant value of $\nu$ \cite{yoo2017fisher}. Then, by using the definition
of Meijer's G-function, $\mathcal{M}\left[F_{\gamma_\mathrm{B}}\left(\gamma_\mathrm{t}\right),1-s\right]$ can be obtained as:
\begin{align}
	&\mathcal{M}\left[F_{\gamma_\mathrm{B}}\left(\gamma_\mathrm{t}\right),1-s\right]=\int_{0}^\infty \gamma_\mathrm{E}^{-s}F_{\gamma_\mathrm{B}}\left(\gamma_\mathrm{t}\right)d\gamma_\mathrm{E},\\
	&\overset{(a)}{=}\int_{0}^\infty \gamma_\mathrm{E}^{-s}G_{1,2}^{1,1}\left(
	\begin{array}{c}
		\frac{\sqrt{\gamma_\mathrm{t}}}{b\sqrt{\bar{\gamma}_\mathrm{B}}}\\
	\end{array}\Bigg\vert
	\begin{array}{c}
		1\\
		a+1,0\\
	\end{array}\right)d\gamma_\mathrm{E},\\\nonumber
	&\overset{(b)}{=}\frac{1}{2\pi j}\int_{\mathcal{L}_2}\int_{0}^\infty \gamma_\mathrm{E}^{-s}\left(\frac{\sqrt{\gamma_\mathrm{t}}}{b\sqrt{\bar{\gamma}_\mathrm{B}}}\right)^{-\zeta}\\
	&\quad\times \frac{\Gamma\left(a+1+\zeta\right)\Gamma\left(-\zeta\right)}{\Gamma\left(1-\zeta\right)}d\zeta d\gamma_\mathrm{E},\label{eq-sop2}
\end{align}
where $(a)$ is obtained by converting the incomplete Gamma
function in \eqref{eq-cdf} into Meijer's G-function, and $(b)$  is achieved by exploiting the definition of Meijer's G-function  and interchanging the integrals order. With the help of \cite[Eq. (3.194.3)]{zwillinger2007table}, the inner integral in \eqref{eq-sop2} can be computed as:
\begin{align}
	&	\int_{0}^\infty \gamma_\mathrm{E}^{-s}\gamma_\mathrm{t}^{-\zeta/2}d\gamma_\mathrm{E}\overset{(c)}{=}\int_{0}^\infty\frac{\gamma_\mathrm{E}^{-s}}{\left(\gamma_\mathrm{E}R_\mathrm{t}+\bar{R}_\mathrm{t}\right)^{\zeta/2}}d\gamma_\mathrm{E},\\
	&\overset{(d)}{=}\bar{R}_\mathrm{t}^{-\zeta/2}\left(\frac{R_\mathrm{t}}{\bar{R}_\mathrm{t}}\right)^{s-1}B\left(1-s,\frac{\zeta}{2}+s-1\right),\\
	&\overset{(e)}{=}\bar{R}_\mathrm{t}^{-\zeta/2}\left(\frac{R_\mathrm{t}}{\bar{R}_\mathrm{t}}\right)^{s-1}\frac{\Gamma\left(1-s\right)\Gamma\left(\frac{\zeta}{2}+s-1\right)}{\Gamma\left(\frac{\zeta}{2}\right)},\label{eq-sop3}
\end{align}
where $(c)$ is obtained representing $\gamma_\mathrm{t}=\gamma_{\mathrm{E}}R_\mathrm{t}+\bar{R}_\mathrm{t}$, $(d)$ is derived
form \cite[Eq. (3.194.3)]{zwillinger2007table}, and $(e)$ is obtained by utilizing the
property of beta function where $B(a_1, b_1)=\frac{\Gamma\left(a_1\right)\Gamma\left(b_1\right)}{
	\Gamma\left(a_1+b_1\right)}$. Now, by inserting \eqref{eq-sop3} into \eqref{eq-sop2}, we have:
\begin{align}\nonumber
	&\mathcal{M}\left[F_{\gamma_\mathrm{B}}\left(\gamma_\mathrm{t}\right),1-s\right]=\frac{\Gamma\left(1-s\right)}{2\pi j}\left(\frac{R_\mathrm{t}}{\bar{R}_\mathrm{t}}\right)^{s-1}\\\nonumber
	&\times\int_{\mathcal{L}_2}\frac{\Gamma\left(a+1+\zeta\right)\Gamma\left(-\zeta\right)\Gamma\left(\frac{\zeta}{2}+s-1\right)}{\Gamma\left(\frac{\zeta}{2}\right)\Gamma\left(1-\zeta\right)}\\
	&\quad\times\left(\frac{\sqrt{\bar{R}_\mathrm{t}}}{b\sqrt{\bar{\gamma}_\mathrm{B}}}\right)^{-\zeta}d\zeta,\label{eq-sop4}
\end{align}
where $\mathcal{L}_2$ is a specific counter. Next, by assuming the change of variable  $u=\frac{\sqrt{\gamma_\mathrm{E}}}{b'\sqrt{\bar{\gamma}_\mathrm{E}}}$ and exploiting the definition of Gamma function, the Mellin transform $\mathcal{M}\left[f_{\gamma_\mathrm{E}}\left(\gamma_\mathrm{E}\right),s\right]$ can be computed as:
\begin{align}\nonumber
	&\mathcal{M}\left[f_{\gamma_\mathrm{E}}\left(\gamma_\mathrm{E}\right),s\right]=\frac{1}{2\bar{\gamma}_\mathrm{E}^{(a'+1)/2}b'^{a'+1}\Gamma(a'+1)}\\\nonumber
	&\times\int_0^\infty
	\gamma_\mathrm{E}^{s-1}\gamma_\mathrm{E}^{(a'-1)/2}\mathrm{e}^{-\frac{\sqrt{\gamma_\mathrm{E}}}{b'\sqrt{\bar{\gamma}_\mathrm{E}}}}d\gamma_{\mathrm{E}}\\
	&=\frac{b'^{2s-2}\bar{\gamma}_\mathrm{E}^{s-1}}{\Gamma\left(a'+1\right)}\Gamma\left(2s+a'-1\right).
	\label{eq-sop5}
\end{align}
Then, by substituting \eqref{eq-sop4} and \eqref{eq-sop5} into \eqref{eq-sop11}, $P_\mathrm{sop}$ can be written as:
\begin{align}\nonumber
	&P_{\mathrm{sop}}=\frac{\bar{R}_\mathrm{t}b'^{-2}}{R_\mathrm{t}\bar{\gamma}_\mathrm{E}\Gamma\left(a'+1\right)}\left(\frac{1}{2\pi j}\right)^2\int_{\mathcal{L}_1}\int_{\mathcal{L}_2}\left(\frac{b\sqrt{\bar{\gamma}_\mathrm{B}}}{\sqrt{\bar{R}_\mathrm{t}}}\right)^{\zeta}\\\nonumber
	&\times\left(\frac{b'^2R_\mathrm{t}\bar{\gamma}_\mathrm{E}}{\bar{R}_\mathrm{t}}\right)^{s} \Gamma\left(\frac{\zeta}{2}+s-1\right) \frac{\Gamma\left(a+1+\zeta\right)\Gamma\left(-\zeta\right)}{\Gamma\left(\frac{\zeta}{2}\right)\Gamma\left(1-\zeta\right)}\\
	&\times\Gamma\left(1-s\right) \Gamma\left(2s+a'-1\right) d\zeta ds.\label{eq-app2}
\end{align}
Eventually, by exploiting the  representation of bivariate Fox's H-function in terms of double contour integral \cite[Eq. (2.56)]{mathai2009h}, the proof is completed.  
\subsection{Proof of Theorem 3}\label{app2}
By inserting \eqref{eq-sc} into \eqref{eq-def-sc}, $\bar{C}_\mathrm{s}$ can be expressed as:
\begin{align}\nonumber
	\bar{C}_\mathrm{s}&=\int_{0}^\infty \log_2\left(1+\gamma_\mathrm{B}\right)f_{\gamma_\mathrm{B}}\left(\gamma_\mathrm{B}\right)F_{\gamma_\mathrm{E}}\left(\gamma_\mathrm{B}\right)d\gamma_\mathrm{B}\\\nonumber
	&+\int_{0}^\infty \log_2\left(1+\gamma_\mathrm{E}\right)f_{\gamma_\mathrm{E}}\left(\gamma_\mathrm{E}\right)F_{\gamma_\mathrm{B}}\left(\gamma_\mathrm{E}\right)d\gamma_\mathrm{E}\\\nonumber
	&-\int_{0}^\infty \log_2\left(1+\gamma_{\mathrm{E}}\right)f_\mathrm{E}\left(\gamma_{\mathrm{E}}\right)d\gamma_{\mathrm{E}}\\
	&=I_1+I_2-I_3.\label{eq-I1}
\end{align}
Next, by substituting the marginal distributions from \eqref{eq-pdf} and \eqref{eq-cdf} into \eqref{eq-I1}, and re-expressing the logarithm function and the incomplete Gamma function in terms of the Meijer's G-function \cite{prudnikov1990more}, $I_1$ can be determined as:
\begin{align}\nonumber
	I_1&=\frac{1}{2\bar{\gamma}_\mathrm{B}^{(a+1)/2}b^{a+1}\Gamma(a+1)\Gamma(a'+1)\ln2}\\\nonumber
	&\times\int_{0}^\infty G_{2,2}^{1,2}\left(
	\begin{array}{c}
		\gamma_\mathrm{B}\\
	\end{array}\Bigg\vert
	\begin{array}{c}
		(1,1)\\
		(1,0)\\
	\end{array}\right)\gamma_\mathrm{B}^{(a-1)/2}\mathrm{e}^{-\frac{\sqrt{\gamma_\mathrm{B}}}{b\sqrt{\bar{\gamma}_\mathrm{B}}}}\\
	&\times G_{1,2}^{1,1}\left(
	\begin{array}{c}
		\frac{\sqrt{\gamma_\mathrm{B}}}{b'\sqrt{\bar{\gamma}_\mathrm{E}}}\\
	\end{array}\Bigg\vert
	\begin{array}{c}
		1\\
		a'+1,0\\
	\end{array}\right)d\gamma_\mathrm{B},\\\nonumber
	&\overset{(f)}{=}\frac{1}{2\bar{\gamma}_\mathrm{B}^{(a+1)/2}b^{a+1}\Gamma(a+1)\Gamma(a'+1)\ln 2}\\\nonumber
	&\times\frac{1}{2\pi j}\int_{\mathcal{L}_1}\frac{\Gamma\left(1-s\right)\Gamma\left(s\right)\Gamma\left(s\right)}{\Gamma\left(1+s\right)}\\
	&\times\underset{J}{\underbrace{\int_{0}^\infty \hspace{-0.3cm}\gamma_\mathrm{B}^{\frac{a-1}{2}+s}\mathrm{e}^{-\frac{\sqrt{\gamma_\mathrm{B}}}{b\sqrt{\bar{\gamma}_\mathrm{B}}}}
			G_{1,2}^{1,1}\left(
			\begin{array}{c}
				\hspace{-0.3cm}\frac{\sqrt{\gamma_\mathrm{B}}}{b'\sqrt{\bar{\gamma}_\mathrm{E}}}\\
			\end{array}\hspace{-0.1cm}\Bigg\vert
			\begin{array}{c}
				1\\
				a'+1,0\\
			\end{array}\hspace{-0.3cm}\right)d\gamma_\mathrm{B}ds}},\label{eq-I11}
\end{align}
where $(f)$ is obtained by using the definition Meijer's G-function. Then, with the help of \cite[Eq. (2.24.3.1)]{prudnikov1990more}, the inner integral $J$ can be computed as:
\begin{align}
	&J=2\left(b\sqrt{\bar{\gamma}_\mathrm{B}}\right)^{2s+a+1} G_{2,2}^{1,2}\left(
	\begin{array}{c}
		\frac{b\sqrt{\bar{\gamma}_\mathrm{B}}}{b'\sqrt{\bar{\gamma}_\mathrm{E}}}\\
	\end{array}\Bigg\vert
	\begin{array}{c}
		\left(-a-2s,1\right)\\
		\left(a'+1,0\right)\\
	\end{array}\right),\label{eq-j}
\end{align}
By plugging \eqref{eq-j} into \eqref{eq-I11} and exploiting the definition of Meijer's G-function, we have:
\begin{align}\nonumber
	&I_1=\frac{1}{\Gamma(a+1)\Gamma(a'+1)\ln 2}\left(\frac{1}{2\pi j}\right)^2\int_{\mathcal{L}_1}\int_{\mathcal{L}_2}\left(b^2\bar{\gamma}_\mathrm{B}\right)^s\\\nonumber
	&\times\left(\frac{b\sqrt{\bar{\gamma}_\mathrm{B}}}{b'\sqrt{\bar{\gamma}_\mathrm{E}}}\right)^\zeta
	\frac{\Gamma\left(1-s\right)\Gamma\left(s\right)\Gamma\left(s\right)\Gamma\left(a'+1-\zeta\right)\Gamma\left(\zeta\right)}{\Gamma\left(1+s\right)\Gamma\left(1+\zeta\right)}\\
	&\times\Gamma\left(1+a+2s+\zeta\right)d\zeta ds,\label{eq-app5}
\end{align}
Consequently, recognizing the
definition of bivariate Fox's H-function, $I_1$ is derived
as:
\begin{align}
	I_1=&\mathcal{H}H_{1,0:2,2:2,2}^{0,1:1,1:1,2}\left(
	\begin{array}{c}
		\frac{b\sqrt{\bar{\gamma}_\mathrm{B}}}{b'\sqrt{\bar{\gamma}_\mathrm{E}}},b^2\bar{\gamma}_\mathrm{B}\\
	\end{array}\Bigg\vert
	\begin{array}{c}
		\mathcal{U}_1\\
		-\\
	\end{array}\Bigg\vert
	\begin{array}{c}
		\mathcal{U}_2\\
		\mathcal{V}_2
	\end{array}\Bigg\vert
	\begin{array}{c}
		\mathcal{U}_3\\
		\mathcal{V}_3
	\end{array}
	\right),
\end{align}
where $\mathcal{H}=\frac{1}{\Gamma(a+1)\Gamma(a'+1)\ln 2}$, $\mathcal{U}_1=		\left(-a;1,2\right)$, $\mathcal{U}_2=\left(1,1\right)$, $\mathcal{U}_3=\left(1,1\right),\left(1,1\right)$, $\mathcal{V}_2=	\left(a'+1,1\right),\left(0,1\right)$, and $\mathcal{V}_3=		\left(1,1\right), \left(0,1\right)$.

Following the same methodology, $I_2$ can be derived as:
\begin{align}
	&I_2=\int_{0}^\infty \log_2\left(1+\gamma_\mathrm{E}\right)f_{\gamma_\mathrm{E}}\left(\gamma_\mathrm{E}\right)F_{\gamma_\mathrm{B}}\left(\gamma_\mathrm{E}\right)d\gamma_\mathrm{E}\\\nonumber
	&=\frac{1}{2\bar{\gamma}_\mathrm{E}^{(a'+1)/2}b'^{a+1}\Gamma(a'+1)\Gamma(a+1)\ln 2}\\\nonumber
	&\times\int_0^\infty G_{2,2}^{1,2}\left(
	\begin{array}{c}
		\gamma_\mathrm{E}\\
	\end{array}\Bigg\vert
	\begin{array}{c}
		(1,1)\\
		(1,0)\\
	\end{array}\right)\gamma_\mathrm{E}^{(a'-1)/2}\mathrm{e}^{-\frac{\sqrt{\gamma_\mathrm{E}}}{b'\sqrt{\bar{\gamma}_\mathrm{E}}}}\\
	&\times G_{1,2}^{1,1}\left(
	\begin{array}{c}
		\frac{\sqrt{\gamma_\mathrm{E}}}{b\sqrt{\bar{\gamma}_\mathrm{B}}}\\
	\end{array}\Bigg\vert
	\begin{array}{c}
		1\\
		a+1,0\\
	\end{array}\right)d\gamma_{\mathrm{E}}\\\nonumber
	&\overset{(g)}{=}\frac{1}{2\bar{\gamma}_\mathrm{E}^{(a'+1)/2}b'^{a+1}\Gamma(a'+1)\Gamma(a+1)\ln 2}\\\nonumber
	&\times\frac{1}{2\pi j}\int_{\mathcal{L}_1}\frac{\Gamma\left(1-s\right)\Gamma\left(s\right)\Gamma\left(s\right)}{\Gamma\left(1+s\right)}\\
	&\times\int_0^\infty \hspace{-0.3cm}G_{1,2}^{1,1}\left(
	\begin{array}{c}
		\hspace{-0.2cm}	\frac{\sqrt{\gamma_\mathrm{E}}}{b\sqrt{\bar{\gamma}_\mathrm{B}}}\\
	\end{array}\Bigg\vert
	\begin{array}{c}
		1\\
		a+1,0\\
	\end{array}\right)\gamma_\mathrm{E}^{\frac{a'-1}{2}+s}\mathrm{e}^{-\frac{\sqrt{\gamma_\mathrm{E}}}{b'\sqrt{\bar{\gamma}_\mathrm{E}}}}d\gamma_\mathrm{E}ds,\\\nonumber
	&\overset{(h)}{=}\frac{1}{\Gamma(a'+1)\Gamma(a+1)\ln 2}\frac{1}{2\pi j}\int_{\mathcal{L}_1}\left(b'^{2}\bar{\gamma}_\mathrm{E}\right)^s\\\nonumber
	&\times\frac{\Gamma\left(1-s\right)\Gamma\left(s\right)\Gamma\left(s\right)}{\Gamma\left(1+s\right)}\\
	&\times G_{2,2}^{1,2}\left(
	\begin{array}{c}
		\frac{b'\sqrt{\bar{\gamma}_\mathrm{E}}}{b\sqrt{\bar{\gamma}_\mathrm{B}}}\\
	\end{array}\Bigg\vert
	\begin{array}{c}
		\left(-a'-2s,1\right)\\
		\left(a+1,0\right)\\
	\end{array}\right)ds\\\nonumber
	&\overset{(i)}{=}\frac{1}{\Gamma(a'+1)\Gamma(a+1)\ln 2}\left(\frac{1}{2\pi j}\right)^2\int_{\mathcal{L}_1}\int_{\mathcal{L}_2}\left(b'^{2}\bar{\gamma}_\mathrm{E}\right)^s\\\nonumber
	&\times\left(\frac{b'\sqrt{\bar{\gamma}_\mathrm{E}}}{b\sqrt{\bar{\gamma}_\mathrm{B}}}\right)^\zeta \frac{\Gamma\left(1-s\right)\Gamma\left(s\right)\Gamma\left(s\right)\Gamma\left(a+1-\zeta\right)\Gamma\left(\zeta\right)}{\Gamma\left(1+s\right)\Gamma\left(1+\zeta\right)}\\
	&\times\Gamma\left(1+a'+2s+\zeta\right)d\zeta ds\\
	&\overset{(j)}{=} \mathcal{H}H_{1,0:1,2:2,2}^{0,1:1,1:1,2}\left(
	\begin{array}{c}
		\frac{b'\sqrt{\bar{\gamma}_\mathrm{E}}}{b\sqrt{\bar{\gamma}_\mathrm{B}}},b^2\bar{\gamma}_\mathrm{E}\\
	\end{array}\Bigg\vert
	\begin{array}{c}
		\mathcal{W}_1\\
		-\\
	\end{array}\Bigg\vert
	\begin{array}{c}
		\mathcal{W}_2\\
		\mathcal{X}_2
	\end{array}\Bigg\vert
	\begin{array}{c}
		\mathcal{W}_3\\
		\mathcal{X}_3
	\end{array}
	\right) \label{eq-app7},
\end{align}
where step $(g)$ is derived by exploiting Meijer's G-function definition, $(h)$ is obtained by computing the inner integral with the help of \cite[Eq. (2.24.3.1)]{prudnikov1990more}, $(i)$ is obtained by using Meijer's G-function definition, and $(j)$ is achieved by exploiting the definition of bivariate Fox's H-function. Besides,  $\mathcal{W}_1=\left(-a';1,2\right)$, $\mathcal{W}_2=\left(1,1\right)$, $\mathcal{W}_3=\left(1,1\right),\left(1,1\right)$, $\mathcal{X}_2=	\left(a+1,1\right),\left(0,1\right)$, and $\mathcal{X}_3=		\left(1,1\right), \left(0,1\right)$.
Furthermore, by re-expressing the logarithm function in terms of Meijer's G-function and then using the integral format provided in \cite[Eq. (2.24.3.1)]{prudnikov1990more}, $I_3$ can be  determined as  
\begin{align}
	I_3&=\int_{0}^\infty \log_2\left(1+\gamma_{\mathrm{E}}\right)f_\mathrm{E}\left(\gamma_{\mathrm{E}}\right)d\gamma_{\mathrm{E}}\\ \nonumber
	&=\frac{1}{2\bar{\gamma}_\mathrm{E}^{(a'+1)/2}b'^{a+1}\Gamma(a'+1)\ln 2}\int_0^\infty \gamma_\mathrm{E}^{(a'-1)/2}\\
	&\times\mathrm{e}^{-\frac{\sqrt{\gamma_\mathrm{E}}}{b'\sqrt{\bar{\gamma}_\mathrm{E}}}} G_{2,2}^{1,2}\left(
	\begin{array}{c}
		\gamma_\mathrm{E}\\
	\end{array}\Bigg\vert
	\begin{array}{c}
		(1,1)\\
		(1,0)\\
	\end{array}\right)\\
	&=\mathcal{O} G_{3,2}^{1,3}\left(
	\begin{array}{c}
		4b'^2\bar{\gamma}_\mathrm{E}\\
	\end{array}\Bigg\vert
	\begin{array}{c}
		(-\frac{a'}{2},1,1)\\
		(1,0)\\
	\end{array}\right),
\end{align}
where $\mathcal{O}=\frac{2^{\frac{a'-2}{2}}}{\sqrt{2\pi}\Gamma\left(a'+1\right)\ln 2}$.
\textcolor{blue}{\subsection{Proof of Proposition 1}\label{app3}
	In ths case of $\bar{\gamma}_\mathrm{B}\rightarrow{\infty}$, the 
	bivariate Fox's H-function in \eqref{eq-app2} is evaluated at the highest poles on the left of $\mathcal{L}_2$, i.e., $\zeta=2(1-s)$. Thus, we have the following results for the counter $\mathcal{L}_2$:
	\begin{align}\nonumber
		\mathcal{R}_2=\,&\frac{1}{2\pi j}\int_{\mathcal{L}_2}
		\Gamma\left(\frac{\zeta}{2}+s-1\right) \\
		&\underset{\chi(\zeta)}{\underbrace{\times\frac{\Gamma\left(a+1+\zeta\right)\Gamma\left(-\zeta\right)}{\Gamma\left(\frac{\zeta}{2}\right)\Gamma\left(1-\zeta\right)}
				\left(\frac{b\sqrt{\bar{\gamma}_\mathrm{B}}}{\sqrt{\bar{R}_\mathrm{t}}}\right)^{\zeta}d\zeta}}\\
		&=\mathrm{Res}\left[\chi(\zeta),2(1-s)\right]\\
		&=\lim_{\zeta\rightarrow{2(1-s)}}\left(\frac{\zeta}{2}+s-1\right)\chi(\zeta)\\
		&\hspace{-1.2cm}=\frac{\Gamma\left(a+3-2s\right)\Gamma\left(2s-2\right)}{\Gamma\left(1-s\right)\Gamma\left(2s-1\right)}\left(\frac{b\sqrt{\bar{\gamma}_\mathrm{B}}}{\sqrt{\bar{R}_\mathrm{t}}}\right)^{2(1-s)}.\label{eq-app3}
	\end{align}
	Now, by applying \eqref{eq-app3} to \eqref{eq-app2}, the asymptotic SOP can be determined as:
	\begin{align}\nonumber
		&\hspace{-0.4cm}P_{\mathrm{sop}}^{\mathrm{asy}}=\frac{b^2\bar{\gamma}_\mathrm{B}}{2\pi j R_\mathrm{t}b'^{2}\bar{\gamma}_\mathrm{E}\Gamma\left(a'+1\right)}\int_{\mathcal{L}_1}\left(\frac{b'^2R_\mathrm{t}\bar{\gamma}_\mathrm{E}}{b^2\bar{\gamma}_\mathrm{B}}\right)^{s}\\
		&\hspace{-0.4cm}\times \underset{\chi'(s)}{\underbrace{\frac{\Gamma\left(2s+a'-1\right)\Gamma\left(a+3-2s\right)\Gamma\left(2s-2\right)}{\Gamma\left(2s-1\right)}}}ds\\ \nonumber
		&=\frac{b^2\bar{\gamma}_\mathrm{B}}{ R_\mathrm{t}b'^{2}\bar{\gamma}_\mathrm{E}\Gamma\left(a'+1\right)}\\
		&\hspace{-0.2cm}\times H_{2,2}^{2,1}\left(
		\begin{array}{c}
			\frac{b^2\bar{\gamma}_\mathrm{B}}{b'^2\bar{\gamma}_\mathrm{E}R_\mathrm{t}}\\
		\end{array}\Bigg\vert
		\begin{array}{c}
			(-a-2,2), (-1,2)\\
			(-2,2), (a'-1,2)\\
		\end{array}\right).\label{eq-app4}
	\end{align}
	Then, by computing the highest pole on the right of the contour $\mathcal{L}_2$, i.e., $s=\frac{a+3}{2}$, 
	\begin{align}
		\hspace{-0.2cm}P_{\mathrm{sop}}^{\mathrm{asy}}=\frac{b^2\bar{\gamma}_\mathrm{B}}{R_\mathrm{t}b'^2\bar{\gamma}_\mathrm{E}\Gamma(a'+1)}\mathrm{Res}\left[\chi'(s),\frac{a+3}{2}\right],
	\end{align}
	\eqref{eq-app4} can be asymptotically simplified as \eqref{eq-asy-sop} and the proof is completed.
	\subsection{Proof of  Proposition 2}\label{app4}
	In the case of $\bar{\gamma}_\mathrm{B}\rightarrow\infty$, the asymptotic ASC can be defined as:
	\begin{align}
		\bar{C}_\mathrm{s}^\mathrm{asy}=I'_1+I'_2-I_3, \label{eq-app8}
	\end{align}
	where $I'_1$ can be determined by computing the bivariate Fox's H-function in \eqref{eq-app5} at the poles $s=0$ and $s=-\frac{1+a+\zeta}{2}$ on the left of the contour $\mathcal{L}_1$. Thus, by using the residue theorem, we have: 
	\begin{align}\nonumber
		&\mathcal{G}_1=\frac{1}{2\pi j}\int_{\mathcal{L}_1}\left(b^2\bar{\gamma}_\mathrm{B}\right)^s\\
		&\times\underset{\upsilon(s)}{\underbrace{\frac{\Gamma(1-s)\Gamma(s)\Gamma(s)}{\Gamma(1+s)}\Gamma\left(1+a+2s+\zeta\right)}}ds\\
		&=\mathrm{Res}\left[\upsilon(s),-\frac{1+a+\zeta}{2}\right]+\mathrm{Res}\left[\upsilon(s),0\right],\label{eq-app6}
	\end{align}
	where
	\begin{align}
		&\mathrm{Res}\left[\upsilon(s),-\frac{1+a+\zeta}{2}\right]\\
		&=\lim_{s\rightarrow -\frac{1+a+\zeta}{2}}\left(s+\frac{1+a+\zeta}{2}\right)\upsilon(s)\\
		&=\frac{\Gamma\left(\frac{3+a+\zeta}{2}\right)\Gamma^2\left(-\frac{1+a+\zeta}{2}\right)}{\Gamma\left(\frac{1-a-\zeta}{2}\right)}\left(b^2\bar{\gamma}_\mathrm{B}\right)^{-\left(\frac{1+a+\zeta}{2}\right)},
	\end{align}
	\begin{align}
		\mathrm{Res}\left[\upsilon(s),0\right]=\Gamma\left(1+a+\zeta\right).
	\end{align}
	Now, by inserting the obtained results into \eqref{eq-app6} and then into \eqref{eq-app5}, we have:
	\begin{align}\nonumber
		&I'_1=\frac{\left(b^2\bar{\gamma}_\mathrm{B}\right)^{-\left(\frac{1+a}{2}\right)}}{2\pi j\Gamma(a+1)\Gamma(a'+1)\ln 2}\int_{\mathcal{L}_2}\left(\frac{1}{b'\sqrt{\bar{\gamma}_\mathrm{E}}}\right)^\zeta\\ \nonumber
		&\times
		\frac{\Gamma\left(a'+1-\zeta\right)\Gamma\left(\zeta\right)\Gamma\left(\frac{3+a+\zeta}{2}\right)\Gamma^2\left(-\frac{1+a+\zeta}{2}\right)}{\Gamma\left(1+\zeta\right)\Gamma\left(\frac{1-a-\zeta}{2}\right)}d\zeta\\ \nonumber
		&+\frac{1}{2\pi j\Gamma(a+1)\Gamma(a'+1)\ln 2}\int_{\mathcal{L}_2}\left(\frac{b\sqrt{\bar{\gamma}_\mathrm{B}}}{b'\sqrt{\bar{\gamma}_\mathrm{E}}}\right)^\zeta\\
		&\times
		\frac{\Gamma\left(a'+1-\zeta\right)\Gamma\left(\zeta\right)\Gamma\left(1+a+\zeta\right)}{\Gamma\left(1+\zeta\right)}d\zeta
	\end{align}
	where, by using the definition of the univariate Fox's H-function, the proof of $I'_1$ (i.e., the first term and second term in \eqref{eq-asy-asc}) is completed. Similarly, $I'_2$ (i.e., the third term and the fourth term in \eqref{eq-asy-asc}) can be obtained by computing the bivariate Fox's H-function in \eqref{eq-app7} at the poles $s=0$ and $s=-\frac{1+a'+\zeta}{2}$. Finally, by plugging $I_3$ into \eqref{eq-app8}, the proof of the asymptotic ASC is completed.}
\bibliographystyle{IEEEtran}
\bibliography{main}




\end{document}